\documentclass[11pt,oneside]{amsart}

\usepackage{amsmath,amssymb}

\theoremstyle{plain} 
\newtheorem{theorem}{Theorem}[section]
\newtheorem{lemma}[theorem]{Lemma}
\newtheorem{corollary}[theorem]{Corollary}
\newtheorem{proposition}[theorem]{Proposition}
\newtheorem{remark}[theorem]{Remark}
\newtheorem{definition}[theorem]{Definition}

\numberwithin{equation}{section}

\begin{document}

\begin{center}

{\bf\large 
Duality between the two-locus Wright-Fisher Diffusion Model and 
the Ancestral Process with Recombination}

\vspace{1cm}

Shuhei Mano

{\it The Institute of Statistical Mathematics,\\
Tachikawa 190-8562, Japan\\
and\\
Japan Science and Technology Agency,\\
Kawaguchi 332-0012, Japan}

\end{center}

\vspace{3cm}

\noindent
Address for correspondence: The Institute of Statistical Mathematics, 
10-3 Midori-cho, Tachikawa, Tokyo 190-8562, Japan; Email: {\tt smano@ism.ac.jp}

\newpage

\noindent
{\bf Abstract}

\noindent

Known results on the moments of the distribution generated by the two-locus
Wright-Fisher diffusion model and a duality between the diffusion process
and the ancestral process with recombination are briefly summarized.
A numerical methods for computing moments by a Markov chain Monte Carlo and
a method to compute closed-form expressions of the moments are presented. 
By using the duality argument properties of the ancestral recombination 
graph are studied in terms of the moments. 

\vspace{1cm}

\noindent
Keywords: duality, diffusion process, ancestral graph, population genetics, 
recombination

\newpage

\section{Introduction}

%
%

In classical population genetics theory behavior of frequency of a gene 
type (allele) has been a central issue (for example, \cite{CrowKimura1971}).
The fate of the allele frequency has been modeled by a diffusion process, 
where the population size is assumed to be sufficiently large. The diffusive
limit, which is called as the Wright-Fisher diffusion model, is expected to 
illustrate actual evolution of the allele frequency in the population. 
Numerical methods for computing likelihood of a sample taken from 
the equilibrium have attracted much interest (for example, \cite{Tavare2004}).
Explicit and closed-form expressions of the whole process have importance
by their own right. Unfortunately, their availability has been limited.
For the one-locus two-allele model without mutation and other evolutionary 
forces closed-form expressions of the probability density of the allele 
frequency at a fixed time were obtained in terms of orthogonal 
polynomials \cite{Malecot1948},\cite{Kimura1955a}. In contrast, for two-locus
models known closed-form expressions have been limited to several moments of 
the distribution generated by the diffusion process \cite{OhtaKimura1969},
\cite{Littler1972}. A comprehensive survey at early 1970's, which is still 
useful, is \cite{Littler1972}. Recently, closed-form expressions of a class 
of moments were obtained in terms of orthogonal polynomials \cite{Mano2005}. 

%
%

The concept of duality has been a powerful tool in stochastic analysis of 
interacting particle systems \cite{Liggett1985}. In population genetics 
theory the moment dual of the Wright-Fisher diffusion model was firstly used 
in \cite{Shiga1981}. A genealogical process of a sample taken from
a population, which is known as the coalescent \cite{Kingman1982}, has been 
useful for population genetic data analyses. The duality was used to obtain
branching-coalescent processes as models of natural selection 
\cite{KroneNeuhauser1997} and conversion bias \cite{Mano2009b}. Ancestral 
processes are process of numbers of ancestral lineages in a section of 
ancestral graphs, which are analogues of the coalescent genealogy. The dual 
of the one-locus two-allele model Wright-Fisher diffusion model with 
directional selection \cite{Kimura1955b} is an ancestral process, which is
the number of ancestral lineages in a section of the ancestral selection graph.
The dual is a birth and death process with linear birth and quadratic death 
rates. It was demonstrated that properties of the birth and death process can 
be studied by referring to classical results on the Wright-Fisher diffusion 
model \cite{Mano2009a}. For the multi-locus model an analogue of 
the coalescent genealogy, called the ancestral recombination graph (ARG), 
was introduced \cite{Griffiths1991}. The two-locus ARG integrates marginal
genealogies at the two loci. The ancestral process, which is a process of 
the numbers of ancestral lineages in a section of the ARG, is the dual of 
the two-locus two-allele Wright-Fisher diffusion model
\cite{EthierGriffiths1990a}. 

In section 2, we briefly summarize known results on the moments of the 
distribution generated by the two-locus two-allele Wright-Fisher diffusion
model. In section 3, the moment duality between the diffusion process and 
the ancestral process, which is a process of the numbers of ancestral lineages
in a section of an ARG, is introduced. A numerical method for computing moments
at a fixed time by a Markov chain Monte Carlo is introduced. In section 4, 
method to compute closed-form expressions of the moments by using an ARG 
terminology is presented. In section 5, by using the duality argument 
properties of the ARG are studied in terms of the moments.

\section{Summary of known results on the moments}

%
%

Consider a random mating monoecious diploid population consisting of $N$ 
individuals. Two linked loci $A$ and $B$ are segregating, where recombination
fraction between the two loci is $r$. Pairs of alleles $A_1, A_2$ and 
$B_1, B_2$ are in the loci $A$ and $B$, respectively. A diffusive limit is 
measuring time in units of $2N$ generations and $2N\rightarrow\infty$, while 
$\rho=4Nr$ is kept constant. Let frequencies of gametes $A_1B_1$,
$A_1B_2$, $A_2B_1$, and $A_2B_2$, be respectively, $x_1$, $x_2$, $x_3$, and 
$1-x_1-x_2-x_3$. Frequencies of the alleles $A_1$ and $A_2$ are denoted by 
$x$ and $1-x$, respectively, and those of the alleles $B_1$ and $B_2$ are 
denoted by $y$ and $1-y$, respectively. Then $x=x_1+x_2$ and $y=x_1+x_3$. 
Set $z=x_1(1-x_1-x_2-x_3)-x_2 x_3$, which is a measure of association 
between $x$ and $y$. The limiting diffusion process 
$\{x_1(t),x_2(t),x_3(t);t\ge 0\}$ is defined in a simplex
\begin{equation*}
K: 0\le x_1 \le x_1+x_2 \le x_1+x_2+x_3\le 1.
\end{equation*}
Let $H=\Phi(K)$, where $\Phi(x_1,x_2,x_3)=(x,y,z)$ is a 
$C^\infty$-diffeomorphism of $K$ onto $H$. The generator of the diffusion
process $\{x(t),y(t),z(t);t\ge 0\}$ in $H$ is \cite{OhtaKimura1969}
\begin{eqnarray}
{\mathcal L}&=&
\frac{x(1-x)}{2}\frac{\partial^2}{\partial x^2}+
\frac{y(1-y)}{2}\frac{\partial^2}{\partial y^2}+
z\frac{\partial^2}{\partial x\partial y}
+z(1-2x)\frac{\partial^2}{\partial x\partial z}
+z(1-2y)\frac{\partial^2}{\partial y\partial z}
\nonumber\\
&&
-z\left(1+\frac{\rho}{2}\right)\frac{\partial}{\partial z}
+\frac{1}{2}\left\{xy(1-x)(1-y)+z(1-2x)(1-2y)-z^2\right\}
\frac{\partial^2}{\partial z^2}.
\label{generator}
\end{eqnarray}

%
%

In the classical population genetics theory some problems of general interests
are concerning events of fixation. The probability of eventual fixation of
an allele and the probability density of the time to the fixation have been 
studied. Some of these properties can be studied in terms of moments of 
the distribution generated by the model by using moment inversion formula.
The probability of eventual fixation of a gamete in the two-locus two-allele 
Wright-Fisher diffusion model governed by the generator (\ref{generator}) is
obtained immediately by the moments. In fact, since the stationary density is 
atomic, $\displaystyle\lim_{t\rightarrow\infty}{\mathbb E}[x(t)y(t)]$ gives 
the fixation probability of the gamete $AB$. For an allele two types of 
fixation can be defined; first fixation occurs when the first of the four 
alleles is lost and the final fixation occurs when an allele at the other 
locus is lost (a gamete fix). These fixation times are
\begin{eqnarray*}
T_1&=&\inf\{t\ge 0; x(t)(1-x(t))y(t)(1-y(t))=0\},\\
T_0&=&\inf\{t\ge 0; x(t)(1-x(t))+y(t)(1-y(t))=0\},
\end{eqnarray*}
respectively. The probability densities are
\begin{eqnarray*}
{\mathbb P}[T_1<t]
&=&\lim_{n\rightarrow\infty}{\mathbb E}
\left[\{1-x(t)(1-x(t))y(y)(1-y(t))\}^n\right],\\
{\mathbb P}[T_0<t]
&=&\lim_{n\rightarrow\infty}{\mathbb E}
\left[\{1-x(t)(1-x(t))-y(y)(1-y(t))\}^n\right],
\end{eqnarray*}
respectively. It seems impossible to obtain explicit and closed-form 
expressions of these limits (see Section 4). Nevertheless, some of the
moments whose closed-form expressions are available are useful to obtain 
upper bounds and approximate formula of these 
probabilities \cite{Littler1972}. By the same reason a closed-form expression 
of the joint distribution of $(x(t),y(t),z(t))$ at a fixed time $t$ is not 
available. 



Let us introduce a classification of the moments of the distribution generated
by the generator (\ref{generator}).

\begin{definition}
The {\it rank} and {\it class} of a moment, which is an expectation of 
a monomial $x^lu^mx_1^n$,\,$l,m,n \in {\mathbb Z}_+$, are $l+m+2n$ and 
$n+\min\{l,m\}$, respectively. The rank is equals to or larger than 
twice of the class.
\end{definition}

\begin{remark}
The class-zero moments have closed-form expressions and they are the moments
of the one-locus Wright-Fisher diffusion model \cite{Kimura1955a}. 
The class-one moments have closed-form expressions \cite{Mano2005} (see below).
\end{remark}

Other moments whose closed-form expressions have been obtained are 
expectations of a type of polynomials:

\begin{lemma}[Lemma 3.6.1 of \cite{Littler1972}]
The manifold of polynomials spanned by the set of polynomials
$\{x^l(1-x)^ly^m(1-y)^mz^n(1-2x)^a(1-2y)^a\}$, where $a=0,1$ and $l,m>0$ if
$n=0$, is closed under the operation of ${\mathcal L}$.
\end{lemma}

\begin{remark}
The polynomials have zero on the boundary of the square $x(1-x)y(1-y)=0$ and
$z=0$. Known closed-form expressions for the moments of polynomials of this 
type are $(l,m,n,a)=(1,1,0,0)$, $(0,0,2,0)$, and $(0,0,0,1)$ by 
\cite{OhtaKimura1969}. Expressions for $(2,1,0,0)$, $(1,2,0,0)$, $(2,0,1,1)$, 
$(0,2,1,1)$, $(2,0,2,0)$, $(0,2,2,0)$, $(2,2,0,0)$, $(1,1,1,1)$, $(1,1,2,0)$,
$(0,0,3,1)$, and $(0,0,4,0)$ are given by \cite{Littler1972},
where the expressions involve eigenvalues whose closed-form expressions are 
not available.
\end{remark}

%
%

In \cite{Mano2005} a closed-form expression of 
${\mathbb E}[x_1(t)|x(t)\in(0,1)]$ was obtained by using a limit of 
a closed-form expression of a class-one moment. It yields an expression of 
the conditional covariance between $x$ and $y$ given that alleles $A_1$ and 
$A_2$ are segregating in the locus $A$. This expression has importance in 
interpreting observable polymorphism in population genetic data analysis. 
Here, we summarize some results in \cite{Mano2005}, because the explicit 
expressions are used in the later sections. If the argument equals to unity, 
the truncated hypergeometric series
\begin{equation*}
y_n(a,b;c;z)=\sum_{i=0}^{n}\frac{(a)_i(b)_i}{(c)_ii!}z^i
\end{equation*}
is expressed by the generalized hypergeometric series~\cite{Bailey1935}
\begin{equation*}
y_n(a,b,c;1)
=\frac{\Gamma(a+n+1)\Gamma(b+n+1)}{n!\Gamma(a+b+n+1)}
{}_3 F_2(a,b,c+n;c,a+b+n+1;1).
\end{equation*}
where ${}_3 F_2(\cdot)$ is the generalized hypergeometric series. A trivial 
but useful identity is
\begin{lemma}[\cite{Mano2005}]
For $m,n\in {\mathbb Z}_+$ and $a,b,c\in {\mathbb C}$,
\begin{eqnarray}
&&\frac{n!\Gamma(a+b+n+1)}{\Gamma(a+n+1)\Gamma(b+n+1)}
y_n(a,b;a+b+m+1;1)\nonumber\\
&&=\frac{m!\Gamma(a+b+m+1)}{\Gamma(a+m+1)\Gamma(b+m+1)}
y_{m}(a,b;a+b+n+1;1).\label{trahyp_id}
\end{eqnarray}
\end{lemma}
\begin{remark}
If $m=0$, it gives an identity
\begin{eqnarray*}
{}_2F_1(a,b;a+b+1;1)&=&{}_3F_2(a,b,a+b+n+1;a+b+1,a+b+n+1;1)\\
&=&\frac{n!\Gamma(a+b+n+1)}{\Gamma(a+n+1)\Gamma(b+n+1)}y_n(a,b,a+b+1;1)\\
&=&\frac{\Gamma(a+b+1)}{\Gamma(a+1)\Gamma(b+1)},
\end{eqnarray*}
which is a spacial case of the Gauss hypergeometric theorem~\cite{Bailey1935}.
\end{remark}

An expression of a power of $p$ in terms of the Gegenbauer polynomial follows
by the orthogonal and complete property \cite{Mano2005}:
\begin{equation}
p^n=\sum_{m=2}^{n+2}2(2m-1)\frac{[n+1]_{m-1}}{(n+1)_{m+1}}
(-1)^mT^1_{m-2}(1-2p),\qquad n\in {\mathbb Z}_+,
\label{pn_id}
\end{equation}
where $T^1_m(\cdot)$ is the Gegenbauer polynomial, which also denoted by 
$C^{(\frac{3}{2})}_1(\cdot)$. $[n]_m$ and $(n)_m$ are falling and rising 
factorials, respectively. By using this expression closed form solutions of
systems of differential equations for class-one moments were obtained. Let 
$\mu_{l,m,n}(t)={\mathbb E}_{pqd}[x(t)^ly(t)^mz(t)^n]$ for 
$l,m,n\in {\mathbb Z}_+$.

\begin{proposition}[(\cite{Mano2005})]
For $n\in {\mathbb Z}_+$,
\begin{equation*}
\mu_{n,0,1}(t)=\sum_{m=2}^{n+2}2(2m-1)\frac{[n+1]_{m-1}}{(n+1)_{m+1}}
(-1)^mT^1_{m-2}(1-2p)de^{-\frac{m(m-1)+\rho}{2}t}.
\end{equation*}
\end{proposition}

\begin{proposition}[\cite{Mano2005}]\label{mu_n10}
For $n\in{\mathbb Z}_+$, 
\begin{equation*}
\mu_{n,1,0}(t)=pq+\frac{2d}{2+\rho}
+\sum_{m=1}^{n-1}E^{(m)}_ne^{-\frac{m(m+1)}{2}t}
+\sum_{m=1}^nF^{(m)}_ne^{-\frac{\rho+m(m+1)}{2}t}
\end{equation*}
except for $\rho=(k+m)(k-m-1),k=m+2,m+3,...,n;m=1,2,...,n$, where
\begin{eqnarray*}
E^{(m)}_n&=&
(-1)^m\frac{[n]_{m+1}}{(n)_{m+1}}
\left[
\frac{2(2m+1)}{m(m+1)}p(1-p)qT^1_{m-1}(1-2p)\right.\nonumber\\
&&\left. +2\left\{
\frac{T^1_m(1-2p)}{2(m+1)+\rho}+\frac{T^1_{m-2}(1-2p)}{2m-\rho}
\right\}d\right],
\end{eqnarray*}
and
\begin{equation}
F^{(m)}_n=
2(-1)^m\frac{[n]_m}{(n)_m}
\left\{\frac{1}{2m+\rho}+\frac{1}{2(m+1)-\rho}\frac{(n-m)(n-m-1)}{(n+m)(n+m+1)}
\right\}T^1_{m-1}(1-2p)d,\label{solF}
\end{equation}
with conventions that the first sum is zero if $n=1$ and $T_{-1}^1(\cdot)=0$.
\end{proposition}

\begin{proof}
A sketch of the proof was given in \cite{Mano2005}. A system of differential 
equations for the moments gives
\begin{equation*}
E^{(m)}_n=\frac{[n]_{m+1}!(2m+1)!}{(n)_{m+1}(m+1)!m!}E^{(m)}_{m+1},
\qquad n\in{\mathbb Z}_+;\, m=1,2,...,n-1,
\end{equation*}
and 
\begin{eqnarray}
&&\left\{(n+m)(n-m-1)-\rho\right\}F^{(m)}_n=
4n(2m+1)\frac{[n-1]_{m-1}}{(n+1)_{m+1}}(-1)^{m+1}T^1_{m-1}(1-2p)d\nonumber\\
&&+n(n-1)F^{(m)}_{n-1},\qquad n\in{\mathbb Z}_+;\, m=1,2,...,n,
\label{eqF}
\end{eqnarray}
with the initial condition
\begin{equation*}
p^n q=pq+\frac{2d}{2+\rho}+\sum_{m=1}^{n-1}E^{(m)}_n+\sum^n_{m=1}F^{(m)}_n,
\qquad n\in {\mathbb Z}_+.
\end{equation*}
By using (\ref{trahyp_id}) with $m=1,2$, it is straightforward to solve 
(\ref{eqF}) for $F^{(m)}_n$ and the solution is (\ref{solF}), except for 
$\rho=(k+m)(k-m-1),k=m+2,m+3,...,n;m=1,2,...,n$ (The exceptional values given 
in \cite{Mano2005} are incorrect). By setting
$E_n^{(m)}=qE_{n,1}^{(m)}(p)+dE_{n,2}^{(m)}(p)$, the initial condition gives
\begin{eqnarray*}
\sum_{m=1}^{n-1}E_{n,1}^{(m)}(p)
&=&p(p-1)\sum_{i=0}^{n-2}p^i\\
&=&p(p-1)
\sum_{i=0}^{n-2}
\sum_{m=2}^{i+2}2(2m-1)\frac{[i+1]_{m-1}}{(i+1)_{m+1}}
(-1)^mT^1_{m-2}(1-2p)\\
&=&p(p-1)\sum_{m=1}^{n-1}
2(-1)^{m+1}\frac{m!(m-1)!}{(2m)!}\\
&&\times y_{n-m-1}(m+1,m,2m+2;1)
T^1_{m-1}(1-2p)\\
&=&p(1-p)\sum_{m=1}^{n-1}
(-1)^{m}\frac{[n]_{m+1}}{(n)_{m+1}}\frac{2(2m+1)}{m(m+1)}T^1_{m-1}(1-2p),
\end{eqnarray*}
for each $n=2,3,...$, where the second equality holds by (\ref{pn_id}) and 
the last equality holds by (\ref{trahyp_id}) with $m=0$. The expression of
summation of $E_{n,2}^{(m)}(p)$ follows by re-arrangement of terms.
\end{proof}

\section{The duality and a numerical method for computing moments}

%
%

The process of the numbers of lineages including non-ancestral lineages (see
below) in a section of a two-locus ARG is a birth and death 
process \cite{Griffiths1991}. When there are $i$ lineages, the birth rate is 
$i\rho/2$ and the death rate is $i(i-1)/2$. The birth and death process is 
identical to the numbers of ancestral lineages in a section of an ancestral 
selection graph, and the moment dual of the birth and death process is 
the Wright-Fisher diffusion of the one-locus two-allele model with directional
selection \cite{Mano2009a}. Note that, an ARG involves gametes whose alleles 
are not ancestral to any allele in a sample. We denote alleles which are not
ancestral to any allele in the sample by $-$, since in principle the allelic 
state of the locus cannot be specified by the sample. For example, a gamete 
$AB$ could be a recombinant descendant of $A-$, which in turn could be 
a recombinant descendant of $--$. The gamete $--$ is involved in the ARG, 
nevertheless, whose alleles are not ancestral to any allele in the sample. 
In this paper, we discuss the ancestral process which generates numbers of 
ancestral lineages in a section of an ARG. The ancestral lineages are a subset
of lineages in an ARG. The stationary distribution of the process with 
infinitely-many-allele mutation model was studied 
by \cite{EthierGriffiths1990b}. A whole graph ${\mathcal G}$ includes marginal 
genealogies ${\mathcal T}_A$ and ${\mathcal T}_B$ at the locus $A$ and $B$, 
respectively. Denote the edges of a graph by $E(\cdot)$. Edges of ${\mathcal G}$ 
are partitioned into 
${\mathcal A}=E({\mathcal G})\cap E({\mathcal T}_A)\cap E({\mathcal T}_B)^c$,
${\mathcal B}=E({\mathcal G})\cap E({\mathcal T}_A)^c\cap E({\mathcal T}_B)$,
${\mathcal C}=E({\mathcal G})\cap E({\mathcal T}_A)\cap E({\mathcal T}_B)$ and
${\mathcal D}=E({\mathcal G})\cap E({\mathcal T}_A)^c\cap E({\mathcal T}_B)^c$.
We call ${\mathcal A}$, ${\mathcal B}$ and ${\mathcal C}$ ancestral lineages. ${\mathcal D}$
are not ancestral lineages since they are not ancestral to any allele in 
a sample. Let ${\mathcal E}_t$ be the edges of a section of $\mathcal G$ taken at 
time $t$ backwards. Denote the number of ancestral lineages by 
$a(t)=|{\mathcal E}_t\cap {\mathcal A}|$, $b(t)=|{\mathcal E}_t\cap {\mathcal B}|$, 
$c(t)=|{\mathcal E}_t\cap {\mathcal C}|$. The marginal transition rates of 
$(a(t),b(t),c(t))$ do not depends on $|{\mathcal E}_t\cap {\mathcal D}|$ and is 
Markovian \cite{EthierGriffiths1990a,Griffiths1991}. The rates are: 
\begin{eqnarray}
(a,b,c)\rightarrow
\left\{
\begin{array}{ll}
(a+1,b+1,c-1), &c\rho/2,\\
(a-1,b-1,c+1), &ab,\\
(a-1,b,c),     &ac+a(a-1)/2,\\
(a,b-1,c),     &bc+b(b-1)/2,\\
(a,b,c-1),     &c(c-1)/2,
\end{array}
\right.
\label{transition}
\end{eqnarray}
with $a+b+c>1$. The backward equation for the joint probability generating 
function $\xi_{l,m,n}(t)={\mathbb E}_{l,m,n}[p^{a(t)}q^{b(t)}g^{c(t)}]$ of 
the Markov chain $\{a(t),b(t),c(t); t\ge 0\}$ on the integer lattice 
${\mathbb Z}_+^3\backslash {\mathbf 0}$ is
\begin{eqnarray}
\frac{d\xi_{l,m,n}}{dt}&=&
-\frac{(l+m+n)(l+m+n-1)+n\rho}{2}\xi_{l,m,n}+\frac{n\rho}{2}\xi_{l+1,m+1,n-1}
\nonumber\\
&&+\frac{l(l-1+2n)}{2}\xi_{l-1,m,n}+\frac{m(m-1+2n)}{2}\xi_{l,m-1,n}
+\frac{n(n-1)}{2}\xi_{l,m,n-1}
\nonumber\\
&&
+lm\xi_{l-1,m-1,n+1},
\label{rec}
\end{eqnarray}
for $(l,m,n)\in{\mathbb Z}_+^3\backslash {\mathbf 0}$, where terms whose 
subscripts have negative integer are zero. It is straightforward to see that
the moments $\nu_{l,m,n}(t)={\mathbb E}_{p,q,g}[x(t)^ly(t)^mx_1(t)^n]$ also 
satisfy the system of the differential equations (\ref{rec}). Therefore 
a moment duality follows immediately \cite{EthierGriffiths1990a}. We give 
a proof, since it is useful to introduce a numerical method for computing 
moments.

\begin{lemma}\label{duality}
The diffusion process $\{x(t),y(t),z(t);t\ge 0\}$ in $H$ with 
$(x(0),y(0),x_1(0))=(p,q,g)$ and the Markov chain $\{a(t),b(t),c(t); t\ge 0\}$
in ${\mathbb Z}_+^3\backslash{\mathbf 0}$ whose transition rates are 
(\ref{transition}) with $(a(0),b(0),c(0))=(l,m,n)$ are dual to each 
other:
\begin{equation*}
{\mathbb E}_{p,q,g}[{x(t)}^l{y(t)}^m{x_1(t)}^n]=
{\mathbb E}_{l,m,n}[p^{a(t)}q^{b(t)}g^{c(t)}].
\end{equation*}
\end{lemma}

\begin{proof}
The system of the differential equations (\ref{rec}) is equivalent to 
an integro-recurrence equation
\begin{equation*}
\xi_{l,m,n}(t)=
\int^t_0{\mathcal T}\xi_{l,m,n}(s)e^{-\gamma_{l,m,n}(t-s)}ds,
\qquad l,m,n\in {\mathbb Z}_+,
\end{equation*}
where 
\begin{eqnarray}
{\mathcal T}\xi_{l,m,n}&=&
\frac{n\rho}{2}\xi_{l+1,m+1,n-1}+\frac{l(l-1+2n)}{2}\xi_{l-1,m,n}+
\frac{m(m-1+2n)}{2}\xi_{l,m-1,n}\nonumber\\
&&+\frac{n(n-1)}{2}\xi_{l,m,n-1}+lm\xi_{l-1,m-1,n+1},
\label{defT}
\end{eqnarray}
and $\gamma_{l,m,n}=((l+m+n)(l+m+n-1)+n\rho)/2$. The integro-recurrence 
equation is recast into
\begin{equation}\label{integrec}
\xi_{l,m,n}(t)=\int^t_0\sum_{l',m',n'}
{\mathbb P}[l'm'n'|l,m,n]\xi_{l',m',n'}(t-s)
\gamma_{l,m,n}e^{-\gamma_{l,m,n}s}ds,
\end{equation}
where the transition probability is given by dividing the rates in 
(\ref{transition}) by $\gamma_{a,b,c}$. Meanwhile
\begin{eqnarray*}
{\mathbb E}_{l,m,n}[p^{a(t)}q^{b(t)}g^{c(t)}]&=&
{\mathbb E}_{l,m,n}\left\{
{\mathbb E}\left[p^{a(t)}q^{b(t)}g^{c(t)}|(a(s),b(s),c(s))=(l',m',n')\right]
\right\}\nonumber\\
&=&
{\mathbb E}_{l,m,n}\left\{
{\mathbb E}_{l',m',n'}
\left[p^{a(t-s)}q^{b(t-s)}g^{c(t-s)}\right]
\right\}\nonumber\\
&=&
{\mathbb E}_{l,m,n}[\xi_{l',m',n'}(t-s)],
\end{eqnarray*}
\end{proof}
where the second equality follows by the strong Markov property. This
expression is equivalent to (\ref{integrec}). Therefore 
$\xi_{l,m,n}(t)={\mathbb E}_{l,m,n}[p^{a(t)}q^{b(t)}g^{c(t)}]$. On the other
hand, it is straightforward by It\^{o}'s formula to see that 
$\nu_{l,m,n}(t)={\mathbb E}_{p,q,g}[x(t)^ly(t)^mx(t)^n]$ satisfy the system of
differential equations (\ref{rec}).

The duality relation is useful for numerical computation of the moments
$\nu_{l,m,n}(t)$ by simulating independent copies of $(a(t),b(t),c(t))$
by the Markov chain Monte Carlo with the transition probabilities 
(\ref{transition}). Consider the simulations stopped at a time $t$. 
The average over $p^{a(t)}q^{b(t)}g^{c(t)}$ of the copies is then an unbiased
estimator of the moment $\nu_{l,m,n}(t)$. A similar method was used for 
computing likelihood of a sample in a varying 
environment \cite{GriffithsTavare1994}. The simulation can be stopped before 
$t$. Let a hitting time be
\begin{equation*}
\tau=\inf\left\{s\ge 0; (a(s),b(s),c(s))\in {\mathcal S}\right\},
\end{equation*}
where ${\mathcal S}=\{(0,0,1),(1,1,0)\}$ is the closed set of states from which
a chain cannot exit. If $\tau<t$,
\begin{eqnarray*}
{\mathbb E}_{l,m,n}[p^{a(t)}q^{b(t)}g^{c(t)}]
&=&{\mathbb E}_{l,m,n}\left\{{\mathbb E}[p^{a(t)}q^{b(t)}g^{c(t)}|
(a(\tau),b(\tau),c(\tau)=(l',m',n')]\right\}\\
&=&{\mathbb E}_{l,m,n}\left\{
{\mathbb E}_{l',m',n'}[p^{a(t-\tau)}q^{b(t-\tau)}g^{c(t-\tau)}]\right\}\\
&=&{\mathbb P}_{l,m,n}[(a(\tau),b(\tau),c(\tau))=(0,0,1)]\nu_{0,0,1}(t-\tau)\\
&&+{\mathbb P}_{l,m,n}[(a(\tau),b(\tau),c(\tau))=(1,1,0)]\nu_{1,1,0}(t-\tau),
\end{eqnarray*}
where the second equality follows by the strong Markov property, and
\begin{equation*}
\nu_{0,0,1}(t-\tau)
=g-\frac{\rho(g-pq)}{2+\rho}\left(1-e^{-\frac{2+\rho}{2}(t-\tau)}\right)
\end{equation*}
and 
\begin{equation*}
\nu_{1,1,0}(t-\tau)
=pq+\frac{2(g-pq)}{2+\rho}\left(1-e^{-\frac{2+\rho}{2}(t-\tau)}\right).
\end{equation*}
From these observations we have a numerical method for computing moments:

\begin{proposition}
Set a Markov time $\sigma=t\wedge\tau$, where $x\wedge y=\min\{x,y\}$. 
An unbiased estimator of $\nu_{l,m,n}(t)$ is the average over following values
obtained by independent copies of $(a(\sigma),b(\sigma),c(\sigma))$ simulated
by the Markov chain Monte Carlo with the transition probabilities 
(\ref{transition}): if $\sigma=t$, the value is 
$p^{a(\sigma)}q^{b(\sigma)}g^{c(\sigma)}$. If $\sigma=\tau$ and 
$(a(\sigma),b(\sigma),c(\sigma))=(0,0,1)$, the value is $\nu_{0,0,1}(t-\tau)$.
If $\sigma=\tau$ and $(a(\sigma),b(\sigma),c(\sigma))=(1,1,0)$, the value is
$\nu_{1,1,0}(t-\tau)$.
\end{proposition}

\section{Closed-form expressions of moments}

%
%

Since the system of differential equations for the moments of the distribution
generated by the two-locus two-allele Wright-Fisher diffusion model 
(\ref{rec}) is that for the joint probability generating function of 
the distribution of the numbers of ancestral lineages in a section of 
a two-locus ARG, relationship among moments can be specified in terms of 
events on the ARG. In (\ref{defT}), the first event is a recombination, 
the second and the third events are marginal coalescence in ${\mathcal T}_A$ and
${\mathcal T}_B$, respectively, and the forth event is a joint coalescence in both
of ${\mathcal T}_A$ and ${\mathcal T}_B$. We call the fifth event as a null 
coalescence, because coalescence events occur in neither ${\mathcal T}_A$ nor 
${\mathcal T}_B$. The following lemma is obvious.
 
\begin{lemma}
The manifold of moments spanned by the set of moments whose ranks and classes
are equals to or smaller than specified values is closed under the operation
of ${\mathcal T}$. Neither of class and rank of a moment change under 
recombination and null coalescence operations.
\end{lemma}

It was shown that the moments $\mu_{l,m,n}(t)$ can be obtained recursively 
from the smaller rank moments \cite{Littler1972}. We present a method to 
compute closed-form expressions of the moments $\nu_{l,m,n}(t)$ by using 
the ARG terminology, since it gives systematic insights in the computation. 
In our approach class of a moment is essential. Since closed-form expressions
of all moments whose classes are less than two are available (see Section 2),
let us start with moments whose classes and ranks are two and $i\,(\ge 4)$, 
respectively: $\nu_{i-4,0,2}$, $\nu_{i-3,1,1}$, $\nu_{i-2,2,0}$,
$\nu_{0,i-4,2}$, $\nu_{1,i-3,1}$, and $\nu_{2,i-2,0}$. The former three and 
the latter three moments are closed respectively by recombinations and null 
coalescences. The expressions for the latter three moments are obtained 
immediately by exchanging $p$ and $q$ in the expressions of the former three 
moments. The system of three differential equations for the former three 
moments whose ranks are $j$ ($4\le j\le i$) are
\begin{eqnarray}
\frac{d}{dt}
\left(
\begin{array}{c}
\nu_{j-4,0,2}\\
\nu_{j-3,1,1}\\
\nu_{j-2,2,0}
\end{array}
\right)&=&
\left(
\begin{array}{ccc}
-\frac{(j-2)(j-3)+2\rho}{2}&\rho&0\\
(j-3)&-\frac{(j-1)(j-2)+\rho}{2}&\frac{\rho}{2}\\
0&2(j-2)&-\frac{j(j-1)}{2}
\end{array}
\right)
\left(
\begin{array}{c}
\nu_{j-4,0,2}\\
\nu_{j-3,1,1}\\
\nu_{j-2,2,0}
\end{array}
\right)\nonumber\\
&&+
\left(
\begin{array}{c}
\frac{(j-1)(j-4)}{2}\nu_{j-5,0,2}\\
\frac{(j-2)(j-3)}{2}\nu_{j-4,1,1}\\
\frac{(j-2)(j-3)}{2}\nu_{j-3,2,0}
\end{array}
\right)+
\left(
\begin{array}{c}
\nu_{j-4,0,1}\\
\nu_{j-3,0,1}\\
\nu_{j-2,1,0}
\end{array}
\right),
\label{de_class2}
\end{eqnarray}
where $\nu_{-1}=0$ by a convention. The solution involves eigenvalues of
the matrix in (\ref{de_class2}), which are roots of the cubic equation
\begin{eqnarray*}
&&\lambda^3+\frac{3j^2-9j+8+3\rho}{2}\lambda^2+
\left\{
\frac{3(j-1)^2(j-2)^2}{2}+(3j^2-11j+15)\rho+\rho^2
\right\}\lambda\\
&&+\frac{j(j-1)^2(j-2)^2(j-3)}{8}
+\frac{(3j^4-22j^3+65j^2-86j+48)\rho+(j^2-5j+8)\rho^2}{4}=0.
\end{eqnarray*}
If $j=4$ the system of differential equations (\ref{de_class2}) involves 
$\nu_{0,0,2}$, $\nu_{1,1,1}$, $\nu_{2,2,0}$ and class one moments, we can 
obtain the closed-form expressions. In fact, the closed-form expressions
of the moments whose classes and ranks are two and four, respectively, were 
obtained by \cite{OhtaKimura1969}. Solving (\ref{de_class2}) iteratively 
in $j=5,6,...,i$, we obtain the closed form expressions of the moments 
whose classes and ranks are two and $i$, respectively. 

Computation of the moments whose classes and ranks are $k\,(\ge 3)$ and 
$i\,(\ge 2k)$, respectively, is in the same manner. To obtain closed-form
expressions of these moments a system of $k+1$ differential equations must
be solved. The solution involves eigenvalues of a tri-diagonal matrix $A_k$
with
\begin{eqnarray*}
&&(A_k)_{s,s-1}=(s-1)(j-2k+s-1),\\
&&(A_k)_{s,s}  =-\frac{(j-k+s-1)(j-k+s-2)+(k-s+1)\rho}{2},\\
&&(A_k)_{s,s+1}=\frac{k-s+1}{2}\rho,
\end{eqnarray*}
for $1\le s\le k+1$ and $2k\le j\le i$. The eigenvalues are roots of 
the $(k+1)$-th degree algebraic equation. If $k\ge 4$ we cannot expect explicit
closed-form expressions of the eigenvalues. The computation eventually 
involves moments whose classes and ranks are $t$ and $u$, respectively, where 
$1\le t\le k$ and $i-k+t\le u\le 2t$. 

\section{Properties of ARG}

%
%

We have been discussed how to compute moments of the distribution generated by
the two-locus two-allele Wright-Fisher diffusion model. The moments are useful 
for studying the two-locus ARG, since the moments of the diffusion, 
$\nu_{l,m,n}(t)$, is the joint probability generating function of 
the distribution of the numbers of ancestral lineages in a section of an ARG of
a sample $(a(0),b(0),c(0))=(l,m,n)$. We define rank and class of the numbers 
of ancestral lineages in a section of an ARG, $(l,m,n)$, by $l+m+2n$ and
$n+\min\{l,m\}$, respectively. An ARG of a class-zero sample is a marginal 
genealogy, whose properties are well known. In the following we consider 
a sample whose class is larger than zero. Since Lemma~\ref{duality} gives
\begin{equation*}
\lim_{t\rightarrow\infty}\nu_{l,m,n}(t)
=\frac{2g}{2+\rho}+\frac{\rho pq}{2+\rho},\qquad n+\min\{l,m\}\ge 1,
\end{equation*}
the stationary distribution of the numbers of ancestral lineages in a section 
of an ARG is
\begin{equation*}
\frac{2}{2+\rho}\delta_{(0,0,1)}+\frac{\rho}{2+\rho}\delta_{(1,1,0)}.
\end{equation*}

%
%

The distribution of the number of ancestral lineages in a section of an ARG 
of a sample $(0,0,2)$ can be obtained from the known closed-form expression 
of the moments of the distribution generated by the two-locus two-allele 
Wright-Fisher diffusion model \cite{OhtaKimura1969}. It seems that general
formula (applicable to all rank moments) of the distribution of a sample 
whose class is larger than one are not available. In contrast, we have general
formula of the distribution of class-one samples. The closed-form expression 
can be obtained by using closed-from expressions of the moments with a finite 
series expansion of the Gegenbauer polynomial \cite{Erdely1953}:
\begin{equation}
T^1_m(1-2p)=\frac{1}{2}
\sum_{i=0}^{m}\frac{(-m)_i(m+1)_{i+2}}{i!(i+1)!}p^i.
\label{exp_gegenbauer}
\end{equation}
Let $\nu_{k,1,0}(t)=\sum_{l,m,n}f_{l,m,n}(t)p^lq^mg^n$, where
$f_{l,m,n}(t)={\mathbb P}_{k,1,0}[(a(t),b(t),c(t))=(l,m,n)]$. The distribution
of a sample $(1,1,0)$ is 
\begin{equation*}
f_{1,1,0}(t)=\frac{\rho}{2+\rho}+\frac{2}{2+\rho}e^{-\frac{2+\rho}{2}t},
\qquad
f_{0,0,1}(t)=\frac{2}{2+\rho}(1-e^{-\frac{2+\rho}{2}t}),
\end{equation*}
since
\begin{equation*}
\nu_{1,1,0}=pq+\frac{2(g-pq)}{2+\rho}(1-e^{-\frac{2+\rho}{2}t}).
\end{equation*}
For samples $(k,1,0)$, $k\ge 2$, from Proposition~\ref{mu_n10} and
(\ref{exp_gegenbauer}) the closed-form expressions have general formula.
For $i\ge 0$ we have
\begin{eqnarray*}
&&f_{i,0,1}(t)=\frac{2}{2+\rho}\delta_{i,0}+g_{i,k}(t)\\
&&+\sum_{m=i}^{k-1}
\frac{(-1)^m[k]_{m+1}}{(k)_{m+1}i!(i+1)!}
\left[
\frac{(-m)_i(m+1)_{i+2}}{2(m+1)+\rho}
+\frac{(2-m)_{i}(m-1)_{i+2}}{2m-\rho}
\right]e^{-\frac{m(m+1)}{2}t},
\end{eqnarray*}
and for $i\ge 2$ we have
\begin{eqnarray*}
&&f_{i,1,0}(t)=-g_{i-1,k}(t)
-\sum_{m=i-1}^{k-1}
\frac{(-1)^m[k]_{m+1}}{(k)_{m+1}(i-1)!i!}
\left\{
(2m+1)(1-m)_{i-2}(m)_i
\right.\\
&&
+\left.
\left[
\frac{(-m)_{i-1}(m+1)_{i+1}}{2(m+1)+\rho}
+\frac{(2-m)_{i-1}(m-1)_{i+1}}{2m-\rho}
\right]
\right\}
e^{-\frac{m(m+1)}{2}t},
\end{eqnarray*}
and
\begin{eqnarray*}
&&f_{1,1,0}(t)=
\frac{\rho}{2+\rho}-g_{0,k}(t)\\
&&+\sum_{m=1}^{k-1}
\frac{(-1)^m[k]_{m+1}}{(k)_{m+1}}
\left[2m+1-
\frac{(m+1)(m+2)}{2(m+1)+\rho}-\frac{(m-1)m}{2m-\rho}
\right]e^{-\frac{m(m+1)}{2}t},
\end{eqnarray*}
where
\begin{eqnarray*}
g_{i,k}(t)&=&\sum_{m=i+1}^{k-1}
\frac{(-1)^m[k]_m}{(k)_m}
\left[
\frac{1}{2m+\rho}+\frac{1}{2(m+1)-\rho}
\frac{(k-m)(k-m-1)}{(k+m)(k+m+1)}
\right]\\
&&\times
\frac{(1-m)_i(m)_{i+2}}{i!(i+1)!}
e^{-\frac{m(m+1)+\rho}{2}t}.
\end{eqnarray*}
We can obtain closed-form expressions of the distribution of samples 
a $(k,0,1)$, $k\ge 1$ in the similar manner.

%
%

Let the waiting times until common ancestors of ${\mathcal T}_A$ and ${\mathcal T}_B$,
respectively, be 
\begin{equation*}
W_A=\inf\{s\ge 0; a(s)+c(s)=1\},\qquad W_B=\inf\{s\ge 0; b(s)+c(s)=1\}.
\end{equation*}

\begin{proposition}
The waiting time until a sample has common ancestor at both of the two loci is 
given by
\begin{eqnarray*}
{\mathbb P}_{l,m,n}\left[W_A\vee W_B\le t\right]&=&
{\mathbb P}_{l,m,n}[(a(t),b(t),c(t))\in{\mathcal S}],
\end{eqnarray*}
where $x\vee y=\max\{x,y\}$. The waiting time until a sample has common 
ancestor at one of the two loci is given by
\begin{equation*}
{\mathbb P}_{l,m,n}\left[W_A\wedge W_B\ge t\right]=
\sum_{n'+\min\{l',m'\}\ge 2}{\mathbb P}_{l,m,n}[(a(t),b(t),c(t))=(l',m',n')].
\end{equation*}
\end{proposition}

\begin{remark}
A recursion of the expectation of $W_A\vee W_B$ for the ARG of a sample 
$(0,0,c)$ is given by Theorem 4 of \cite{Griffiths1991}. Theorem 5 
of \cite{Griffiths1991} gives a closed-form expression of the joint Laplace 
transform of $W_A\vee W_B$ and $W_A\wedge W_B$ for the ARG of a sample 
$(0,0,2)$.
\end{remark}

%
%

The idea of the number of recombination events in a sample was introduced by 
\cite{HudsonKaplan1985}. The number of recombination events on the two-locus 
ARG including non-ancestral lineages was considered by 
\cite{EthierGriffiths1990a,EthierGriffiths1990b}, and a closed-form expression
of the probability generating function of the number of recombination events 
was given. Here, we consider the number of recombination events on ancestral 
lineages of an ARG. Let $s(t)$ be the number of recombination events occurring 
to ${\mathcal C}$ lineages of an ARG in a time interval $(0,t)$. 
The recombination events are subset of the recombination events occurred on
whole lineages of the ARG.

\begin{lemma}\label{jointpgf}
The joint probability generating function of $(a(t),b(t),c(t),s(t))$ is
\begin{equation*}
{\mathbb E}_{l,m,n,0}[p^{a(t)}q^{b(t)}g^{c(t)}v^{s(t)}]
={\mathbb E}_{l,m,n,0}\left[p^{a_v(t)}q^{b_v(t)}g^{c_v(t)}
\exp\left\{-\frac{\rho(1-v)}{2}\int_0^t c_v(u)du\right\}\right],
\end{equation*}
where $\{a_v(t),b_v(t),c_v(t);t\ge 0\}$ is a modified process of
$\{a(t),b(t),c(t);t\ge 0\}$ in which the recombination fraction is $rv$,
where $0\le v\le 1$.
\end{lemma}

\begin{proof}
Let $\zeta_{l,m,n}(t)={\mathbb E}_{l,m,n,0}[p^{a(t)}q^{b(t)}g^{c(t)}v^{s(t)}]$.
For $(l,m,n)\in {\mathbb Z}_+^3\backslash{\bf 0}$, we have
\begin{eqnarray}
\frac{d\zeta_{l,m,n}}{dt}&=&
-\frac{(l+m+n)(l+m+n-1)+nv\rho}{2}\zeta_{l,m,n}
+\frac{nv\rho}{2}\zeta_{l+1,m+1,n-1}\nonumber\\
&&+\frac{l(l-1+2n)}{2}\zeta_{l-1,m,n}+\frac{m(m-1+2n)}{2}\zeta_{l,m-1,n}
+\frac{n(n-1)}{2}\zeta_{l,m,n-1}\nonumber\\
&&+lm\zeta_{l-1,m-1,n+1}-\frac{n(1-v)\rho}{2}\zeta_{l,m,n}
\label{rec2}
\end{eqnarray}
with the initial condition $\xi_{l,m,n}(0)=p^lq^mg^n$. This is uniquely 
solved by means of the Feynman-Kac formula and the result is 
Lemma~\ref{jointpgf}.
\end{proof}

\begin{theorem}
The conditional probability generating function of the number of ${\mathcal A}$
lineages in a section of an ARG of a sample $(n,0,1)$ given that no 
recombination events occur in a time interval $(0,t)$ is \begin{equation*}
{\mathbb E}_{n,0,1,0}[p^{a(t)}|s(t)=0]=\tilde{\nu}_{n,0,1}(t),
\end{equation*}
where $\tilde{\nu}_{n,0,1}(t)$ is $\nu_{n,0,1}(t)$ with setting $q=g=1$ and 
$\rho=0$.
\end{theorem}
\begin{proof}
By (\ref{transition}) we see that $b_0(t)=0$ and $c_0(t)=1$ for all $t$,
and the marginal process $\{a_0(t);t\ge 0\}$ is a death process with death 
rate $i(i+1)/2$ when $a_0(t)=i$. The joint probability generating function of
$(a(t),b(t),c(t))$ given that no recombination events occur in a time interval
$(0,t)$ is  
\begin{eqnarray*}
\lim_{v\rightarrow 0}{\mathbb E}_{n,0,1,0}[p^{a(t)}q^{b(t)}g^{c(t)}v^{s(t)}]
&=&{\mathbb E}_{n,0,1,0}[p^{a(t)}q^{b(t)}g^{c(t)},s(t)=0]\\
&=&{\mathbb E}_{n,0,1,0}\left[p^{a_0(t)}q^{b_0(t)}g^{c_0(t)}
\exp\left\{-\frac{\rho}{2}\int_0^t c_0(u)du\right\}\right]\\
&=&g{\mathbb E}_{n,0,1,0}[p^{a_0(t)}]e^{-\frac{\rho}{2}t},
\end{eqnarray*}
where the first equality follows by Lebesgue's convergence theorem and
the second equality follows by Lemma~\ref{jointpgf}. Setting $p=q=g=1$,
we have ${\mathbb P}_{n,0,1,0}[s(t)=0]=e^{-\rho t/2}$, while setting
$q=g=1$ we have ${\mathbb E}_{n,0,1,0}[p^{a(t)},s(t)=0]
=\tilde{\nu}_{n,0,1}(t)e^{-\frac{\rho}{2}t}$,
where $\tilde{\nu}_{n,0,1}(t)={\mathbb E}_{n,0,1,0}[p^{a_0(t)}]$.
\end{proof}

\begin{remark}
Explicit closed-form expression of $\tilde{\nu}_{n,0,1}(t)$ is available in
Section 2, since $\nu_{n,0,1}(t)=\mu_{n+1,1,0}(t)+\mu_{n,0,1}(t)$. This
expression follows immediately by considering the ARG. Recombination might
occur on the single ${\mathcal C}$ lineage. By Poisson nature of recombination events,
the probability that no recombination occur on the single lineage is 
$e^{-\frac{\rho}{2}t}$. The marginal process $\{a_0(t); t\ge 0\}$ follows 
the death process independently.
\end{remark}

\begin{lemma}\label{numrec}
Let ${\mathcal S}$ be the absorbing states. The probability generating 
function of the number of recombination events on ancestral lineages of 
an ARG until a sample has common ancestor at both of the two loci is
\begin{equation*}
{\mathbb E}_{l,m,n,0}[v^{s(\tau)}]
={\mathbb E}_{l,m,n,0}\left[
\exp\left\{-\frac{\rho(1-v)}{2}\int_0^{\tau_v} c_v(u)du\right\}\right],
\end{equation*}
where $\tau_v=\inf\{s\ge 0; (a_v(s),b_v(s),c_v(s))={\mathcal S}\}$.
\end{lemma}

\begin{proof}
Let $\zeta_{l,m,n}={\mathbb E}_{l,m,n,0}[p^{a(\tau)}q^{b(\tau)}g^{c(\tau)}
v^{s(\tau)}]$. For $(l,m,n)\in {\mathbb Z}_+^3\backslash{\bf 0}$, we have
\begin{eqnarray}
0&=&
-\frac{(l+m+n)(l+m+n-1)+nv\rho}{2}\zeta_{l,m,n}
+\frac{nv\rho}{2}\zeta_{l+1,m+1,n-1}\nonumber\\
&&+\frac{l(l-1+2n)}{2}\zeta_{l-1,m,n}+\frac{m(m-1+2n)}{2}\zeta_{l,m-1,n}
+\frac{n(n-1)}{2}\zeta_{l,m,n-1}\nonumber\\
&&+lm\zeta_{l-1,m-1,n+1}-\frac{n(1-v)\rho}{2}\zeta_{l,m,n}
\label{rec3}
\end{eqnarray}
with the boundary condition $\xi_{0,0,1}=g$ and $\xi_{1,1,0}=pq$. This 
boundary value problem is uniquely solved by means of the Feynman-Kac formula.
That is
\begin{eqnarray*}
&&\zeta_{l,m,n}\\
&&=g{\mathbb E}_{l,m,n,0}\left[
\exp\left\{-\frac{\rho(1-v)}{2}\int_0^{\tau_v} c_v(u)du\right\},
(a_v(\tau_v),b_v(\tau_v),c_v(\tau_v))=(0,0,1)\right]\\
&&+pq{\mathbb E}_{l,m,n,0}\left[
\exp\left\{-\frac{\rho(1-v)}{2}\int_0^{\tau_v} c_v(u)du\right\},
(a_v(\tau_v),b_v(\tau_v),c_v(\tau_v))=(1,1,0)\right].
\end{eqnarray*}
On the other hand we have
\begin{eqnarray*}
{\mathbb E}_{l,m,n,0}[p^{a(\tau)}q^{b(\tau)}g^{c(\tau)}v^{s(\tau)}]
&=&g{\mathbb E}_{l,m,n,0}[v^{s(\tau)},(a(\tau),b(\tau),c(\tau))=(0,0,1)]\\
&&+pq{\mathbb E}_{l,m,n,0}[v^{s(\tau)},(a(\tau),b(\tau),c(\tau))=(1,1,0)].
\end{eqnarray*}
Thus the probability generating functions of the number of recombination 
events on ancestral lineages of an ARG until a sample has common ancestor at 
both of the two loci with the given state in which the sample path absorbed are
\begin{eqnarray*}
&&{\mathbb E}_{l,m,n,0}[v^{s(\tau)},(a(\tau),b(\tau),c(\tau))=(0,0,1)]\\
&&={\mathbb E}_{l,m,n,0}\left[
\exp\left\{-\frac{\rho(1-v)}{2}\int_0^{\tau_v} c_v(u)du\right\},
(a_v(\tau_v),b_v(\tau_v),c_v(\tau_v))=(0,0,1)\right],
\end{eqnarray*}
and
\begin{eqnarray*}
&&{\mathbb E}_{l,m,n,0}[v^{s(\tau)},(a(\tau),b(\tau),c(\tau))=(1,1,0)]\\
&&={\mathbb E}_{l,m,n,0}\left[
\exp\left\{-\frac{\rho(1-v)}{2}\int_0^{\tau_v} c_v(u)du\right\},
(a_v(\tau_v),b_v(\tau_v),c_v(\tau_v))=(1,1,0)\right].
\end{eqnarray*}
Lemma~\ref{numrec} follows as the summation of these two probability 
generating functions.
\end{proof}

\begin{corollary}\label{numrec_exp}
The expected number of recombination events on an ancestral lineages of
an ARG until a sample has a common ancestor at both of the two loci is
\begin{eqnarray*}
{\mathbb E}_{l,m,n,0}[s(\tau)]
&=&\frac{\rho}{2}{\mathbb E}_{l,m,n,0}\left[\int_0^{\tau}c(u)du\right]\\
&=&\frac{\rho}{2}
\sum_{(l',m',n')\in {\mathbb Z}_+^3\backslash\{{\mathbf 0},{\mathcal S}\}}
\int_0^{\infty}k{\mathbb P}_{l,m,n,0}[(a(u),b(u),c(u))=(l',m',n')]du.
\end{eqnarray*}
\end{corollary}

\begin{proof}
Let $T_{l',m',n'}$ be the sojourn time of a sample path of the process of 
the numbers of ancestral lineages in a section of an ARG of a sample $(l,m,n)$
stays at a state $(l',m',n') \notin {\mathcal S}$. Then
\begin{eqnarray*}
&&{\mathbb E}_{l,m,n,0}\left[\int_0^{\tau}c(u)du\right]
=\sum_{(l',m',n')\in{\mathbb Z}_+^3\backslash\{{\mathbf 0},{\mathcal S}\}}
n'{\mathbb E}_{l,m,n,0}\left[T_{l',m',n'}\right]\\
&&=\sum_{(l',m',n')\in{\mathbb Z}_+^3\backslash\{{\mathbf 0},{\mathcal S}\}}
n'{\mathbb E}_{l,m,n,0}\left[
\int_0^{\infty}I_{(l',m',n')}(a(t),b(t),c(t))\right]\\
&&=\sum_{(l',m',n')\in {\mathbb Z}_+^3\backslash\{{\mathbf 0},{\mathcal S}\}}
n'\int_0^{\infty}{\mathbb P}_{l,m,n,0}[(a(u),b(u),c(u))=(l',m',n')]du,
\end{eqnarray*}
where the last equality follows by Fubini's theorem.
\end{proof}

\begin{remark}
A recursion of ${\mathbb E}_{l,m,n,0}[s(\tau)]$ is given by Theorem 6 of
\cite{Griffiths1991}.
\end{remark}

\begin{corollary}\label{norec}
The probability that no recombination events occur on an ARG until a sample 
has common ancestor at both of the two loci is
\begin{equation*}
{\mathbb P}_{l,m,n,0}[s(\tau)=0]=
{\mathbb E}_{l,m,n,0}\left[
\exp\left\{-\frac{\rho}{2}\int_0^{\tau_0} c_0(u)du\right\}\right].
\end{equation*} 
\end{corollary}

\begin{theorem}\label{norec_00n}
The probability that no recombination events occur on an ARG of a sample 
$(0,0,n)$ until the sample has common ancestor at both of the two loci is
\begin{equation*}
\frac{(n-1)!}{(\rho+1)_{n-1}}.
\end{equation*}
\end{theorem}

\begin{proof}
By (\ref{transition}) we see that $a_0(t)=b_0(t)=0$ for all $t$ and 
the marginal process $\{c_0(t); t\ge 0\}$ is a death process with death rate 
$i(i-1)/2$ when $c_0(t)=i$. Let a hitting time be 
$\gamma=\inf\{s\ge 0; a(s)=n-1\}$. By Corollary~\ref{norec},
\begin{eqnarray*}
{\mathbb P}_{0,0,n,0}[s(\tau)=0]
&=&{\mathbb E}_{0,0,n,0}
\left[\exp\left\{-\frac{\rho}{2}\int_0^{\tau_0}c_0(u)du\right\}\right]\\
&=&{\mathbb E}_{0,0,n,0}\left\{
{\mathbb E}
\left[
\left.\exp\left\{-\frac{\rho}{2}\int_0^{\tau_0}c_0(u)du\right\}
\right|\gamma\right]\right\}\\
&=&{\mathbb E}_{0,0,n,0}
\left\{{\mathbb E}_{0,0,n-1,0}
\left[\exp\left\{-\frac{\rho}{2}\int_0^{\tau_0}c_0(u)du\right\}\right]
e^{-\frac{n\rho}{2}\gamma}\right\}
\\
&=&
{\mathbb E}_{0,0,n-1,0}
\left[
\exp\left\{-\frac{\rho}{2}\int_0^{\tau_0}c_0(u)du\right\}\right]
\frac{n-1}{n-1+\rho},
\end{eqnarray*}
where the third equality follows by the strong Markov property, with
the boundary condition 
\begin{equation*}
{\mathbb E}_{0,0,1,0}
\left[
\exp\left\{-\frac{\rho}{2}\int_0^{\tau_0}c_0(u)du\right\}\right]=1.
\end{equation*}
The recursion solved immediately and we get Theorem~\ref{norec_00n}.
\end{proof}

\begin{theorem}
The probability that no recombination events occur on an ARG of a sample 
$(n,0,1)$ until the sample has common ancestor at both of the two loci is
\begin{equation*}
\prod_{i=0}^n\frac{i(i+1)}{i(i+1)+\rho}.
\end{equation*}
\end{theorem}

\begin{proof}
By (\ref{transition}) we see that $b_0(t)=0$ and $c_0(t)=1$ for all $t$,
and the marginal process $\{a_0(t);t\ge 0\}$ is a death process with death 
rate $i(i+1)/2$ when $a_0(t)=i$. By Corollary~\ref{norec}, we have
\begin{equation*}
{\mathbb P}_{n,0,1}[s(\tau)=0]={\mathbb E}_{n}[e^{-\frac{\rho}{2}\tau_0}].
\end{equation*}
The expression follows by a similar argument as used in Proof of 
Theorem~\ref{norec_00n}.
\end{proof}

Finally, let us consider the limit $\rho\rightarrow\infty$. For the purpose
we introduce two processes. The one is a diffusion process 
$\{x_{\infty}(t),y_{\infty}(t);t\ge 0\}$ in $[0,1]^2$ with a generator
\begin{equation*}
{\mathcal L_{\infty}}=\frac{x(1-x)}{2}\frac{\partial^2}{\partial x^2}+
\frac{y(1-y)}{2}\frac{\partial^2}{\partial y^2}
\end{equation*}
and $(x_{\infty}(0),y_{\infty}(0))=(p,q)$. The other is a Markov chain 
$\{a_{\infty}(t),b_{\infty}(t);t\ge 0\}$ in 
${\mathbb Z}_+^2\backslash{\mathbf 0}$ whose transition rates are:
\begin{eqnarray*}
(a,b)\rightarrow
\left\{
\begin{array}{ll}
(a-1,b),     &a(a-1)/2,\\
(a,b-1),     &b(b-1)/2.\\
\end{array}
\right.
\end{eqnarray*}
and $(a_{\infty}(0),b_{\infty}(0))=(l,m)$. Let 
$\tau_{\infty}=\{s\ge 0; (a_{\infty}(s),b_{\infty}(s))
\in\{(1,0),(0,1),(1,1)\}\}$.

\begin{theorem}[Theorem 1 and 2 of \cite{Ethier1979}]\label{limit}
If $(x_{\infty}(0),y_{\infty}(0))=(p,q)$, then $\{x(t),y(t);t\ge 0\}$ 
converges weakly in $C([0,\infty),[0,1]^2)$ to
$\{x_{\infty}(t),y_{\infty}(t);t\ge 0\}$, and
$\{z(t)-de^{-\frac{\rho}{2}t};t\ge 0\}$ converges weakly in 
$C([0,\infty),{\mathbb R})$ to the zero process in ${\mathbb R}$ as 
$\rho\rightarrow\infty$. The two function spaces are given the topology of
uniform convergence on compact intervals.
\end{theorem}

\begin{corollary}\label{limit_pgf}
The probability generating function of the distribution of the numbers of 
ancestral lineages in a section of an ARG of a sample 
$(a(0),b(0),c(0))=(l,m,n)$ has a limit
\begin{equation*}
{\mathbb E}_{l,m,n}[p^{a(t)}q^{b(t)}g^{c(t)}]
\rightarrow 
{\mathbb E}_{l+n}[p^{a_{\infty}(t)}]{\mathbb E}_{m+n}[q^{b_{\infty}(t)}],
\qquad\rho\rightarrow\infty.
\end{equation*}
\end{corollary}

\begin{proof}
From Lemma~\ref{duality} we have
\begin{eqnarray*}
{\mathbb E}_{l,m,n}[p^{a(t)}q^{b(t)}g^{c(t)}]
&=&{\mathbb E}_{p,q,g}[x(t)^{l+n}y(t)^{m+n}]\\
&&+\sum_{i=1}^n
\frac{n!}{(n-i)!i!}
{\mathbb E}_{p,q,g}[x(t)^{l+n-i}y(t)^{m+n-i}z(t)^i].
\end{eqnarray*}
It follows from Theorem~\ref{limit} and Lebesgue's convergence theorem that
\begin{equation*}
{\mathbb E}_{p,q,g}[x(t)^{l+n-i}y(t)^{m+n-i}z(t)^i]
\le
{\mathbb E}_{p,q,g}[z(t)^i]\rightarrow 0 \qquad 
\rho\rightarrow\infty,
\end{equation*}
for $t>0$, $i=1,2,...,n$, while
\begin{equation*}
{\mathbb E}_{p,q,g}[x(t)^{l+n}y(t)^{m+n}]
\rightarrow
{\mathbb E}_{p}[x_{\infty}(t)^{l+n}]{\mathbb E}_{q}[y_{\infty}(t)^{m+n}]
={\mathbb E}_{l+n}[p^{a_{\infty}(t)}]{\mathbb E}_{m+n}[q^{b_{\infty}(t)}].
\end{equation*}
as $\rho\rightarrow\infty$. The last equality is a result of the duality
between $\{x_{\infty}(t);t\ge 0\}$ and $\{a_{\infty}(t);t\ge 0\}$, and 
between $\{y_{\infty}(t);t\ge 0\}$ and $\{b_{\infty}(t);t\ge 0\}$.
\end{proof}

Corollary~\ref{limit_pgf} shows that all $AB$ gametes in a sample 
instantaneously split into a pair $A-$ and $-B$ gametes in the limit 
$\rho\rightarrow\infty$. Therefore the length of ${\mathcal C}$ lineages
in an ARG goes to zero in this limit. 

\begin{theorem}
The expected lengths of ${\mathcal C}$ lineages and whole lineages of an ARG
of a sample $(l,m,n)$ until the sample has common ancestor at both of the 
two loci are
\begin{equation}
{\mathbb E}_{l,m,n}\left[\int_0^\tau c(u)du\right]
\rightarrow\frac{2}{\rho}{\mathbb E}_{l,m}
\left[\int_0^{\tau_{\infty}}a_{\infty}(u)b_{\infty}(u)du\right]
+\frac{2n}{\rho},\label{limit_c}
\end{equation}
and
\begin{equation}
{\mathbb E}_{l,m,n}\left[\int_0^\tau (a(u)+b(u)+c(u))du\right]
\rightarrow{\mathbb E}_{l,m}
\left[\int_0^{\tau_{\infty}}(a_{\infty}(u)+b_{\infty}(u))du\right],
\label{limit_all}
\end{equation}
respectively.
\end{theorem}

\begin{proof}
Let 
$\displaystyle\eta_{l,m,n}=\lim_{\rho\rightarrow\infty}
{\mathbb E}_{l,m,n,0}[s(\tau)=0]$ and $\lambda_{l,m}=\eta_{l,m,0}$.
For $(l,m,n)\in {\mathbb Z}_+^3\backslash{\mathbf 0}$ we have
$\eta_{l,m,n}=n+\lambda_{l+n,m+n}$ for $n\ge 1$ and
\begin{equation*}
0=lm-\frac{l(l-1)+m(m-1)}{2}\lambda_{l,m}
+\frac{l(l-1)}{2}\lambda_{l-1,m}+\frac{m(m-1)}{2}\lambda_{l,m-1},
\end{equation*}
with the boundary condition $\lambda_{1,0}=\lambda_{0,1}=\lambda_{1,1}=0$. 
This boundary value problem is uniquely solved by means of 
the Feynman-Kac formula. That is
\begin{equation*}
\lambda_{l,m}={\mathbb E}_{l,m}
\left[\int_0^{\tau_{\infty}}
a_{\infty}(u)b_{\infty}(u)du\right]
\end{equation*}
From Corollary~\ref{numrec_exp}, we have (\ref{limit_c}). Let 
\begin{equation*}
\eta'_{l,m,n}=\lim_{\rho\rightarrow\infty}
{\mathbb E}_{l,m,n,0}\left[\int_0^{\tau_{\infty}}(a(u)+b(u)+c(u))du\right]
\end{equation*}
and $\lambda'_{l,m}=\eta'_{l,m,0}$. For 
$(l,m,n)\in {\mathbb Z}_+^3\backslash{\mathbf 0}$ we have
$\eta'_{l,m,n}=\lambda'_{l,m}$ for $n\ge 1$ and
\begin{equation*}
0=l+m-\frac{l(l-1)+m(m-1)}{2}\lambda'_{l,m}
+\frac{l(l-1)}{2}\lambda'_{l-1,m}+\frac{m(m-1)}{2}\lambda'_{l,m-1},
\end{equation*}
with the boundary condition $\lambda'_{1,0}=\lambda'_{0,1}=\lambda'_{1,1}=0$.
This boundary problem is also solved and we have (\ref{limit_all}).

\end{proof}



\begin{thebibliography}{99}

\footnotesize

\bibitem{Bailey1935}
{\sc Bailey, W.~N.} (1935). {\em Generalised Hypergeometric Series.} 
Cambridge University Press, Cambridge.

\bibitem{CrowKimura1971}
{\sc Crow, J.~F. and Kimura, M.} (1971).
{\em Introduction to Population Genetics Theory.} Harper and low, New York.

\bibitem{Erdely1953}
{\sc Erd\'{e}ly, A.} ed. (1953).
{\em Higher Transcendental Functions, Vol. I.} McGrow-Hill, New York.

\bibitem{Ethier1979}
{\sc Etheir, S.~N.} (1979).
A limit theorem for two-locus diffusion models in population genetics. 
{\em J. Appl. Probab.} {\bf 16,} 402--408.

\bibitem{EthierGriffiths1990a}
{\sc Etheir, S.~N. and Griffiths, R.~C.} (1990a).
The neutral two-locus model as a measure-valued diffusion.
{\em Adv. Appl. Prob.} {\bf 22,} 773--786.

\bibitem{EthierGriffiths1990b}
{\sc Etheir, S.~N. and Griffiths, R.~C.} (1990b).
On the two-locus sampling distribution. 
{\em J. Math. Biol.} {\bf 29,} 131--159.

\bibitem{Griffiths1991}
{\sc Griffiths, R.~C.} (1991). 
The two-locus ancestral graph. In {\em Selected proceedings of the symposium 
on applied probability, Sheffield, 1989.} ed. I.V. Basawa and R.L. Taylor.
IMS Lecture notes--monograph series, 18, pp. 100--117.

\bibitem{GriffithsTavare1994}
{\sc Griffiths, R.~C. and Tavar'{e}, R.~C.} (1994).
Sampling theory for neutral alleles in a varying environment.
{\em Phil. Trans. R. Soc. Lond.} {\bf B 344,} 403--410.

\bibitem{HudsonKaplan1985}
{\sc Hudson, R.~R. and Kaplan, N.~L.} (1985).
Statistical properties of the number of recombination events in the history 
of a sample of DNA sequences. {\em Genetics} {\bf 111,} 147-164.

\bibitem{Kimura1955a}
{\sc Kimura, M.}(1955).
Solution of a process of random genetic drift with a continuous model.
{\em Proc. Acad. Natul. Sci. USA} {\bf 41,} 144--150.

\bibitem{Kimura1955b}
{\sc Kimura, M.}(1955).
Stochastic process and distribution of gene frequencies under natural
selection. 
{\em Cold Spring Harbor Symposia on Quantitative Biology} {\bf 20,} 33--53.

\bibitem{Kingman1982}
{\sc Kingman, J.~F.~C.} (1982).
The coalescent.
{\em Stoch. Proc. Appl.} {\bf 13,} 235--248.

\bibitem{KroneNeuhauser1997}
{\sc Krone, S.~M. and Neuhauser, C.} (1997).
Ancestral process with selection. 
{\em Theor. Popul. Biol.} {\bf 51,} 210--237.

\bibitem{Liggett1985}
{\sc Liggett, T.~M.} (1985).
{\em Interacting Particle Systems.} Springer-Verlag, Berlin. 

\bibitem{Littler1972}
{\sc Littler, R.~A.} (1972).
Multidimensional stochastic models in genetics. Ph.D. Thesis. Monash 
University.

\bibitem{Malecot1948}
{\sc Mal\'{e}cot, G.} (1948)
{\em Les Mathematiques de l'h\'{e}r\'{e}dit\'{e}.} Masson et Cie, 
{\'{E}diteurs}.

\bibitem{Mano2005}
{\sc Mano, S.} (2005).
Random Genetic Drift and Gamete Frequency.
{\em Genetics} {\bf 171,} 2043--2050.

\bibitem{Mano2009a}
{\sc Mano, S.} (2009a).
Duality, ancestral and diffusion processes in models with selection.
{\em Theor. Popul. Biol.} {\bf 75,} 164--175.

\bibitem{Mano2009b}
{\sc Mano, S.} (2009b).
Ancestral process with bias in gene conversion. arXiv:0907.1127

\bibitem{OhtaKimura1969}
{\sc Ohta, T. and Kimura, M.} (1969).
Linkage disequilibrium due to random genetic drift. 
{\sc Genet. Res.} {\bf 13,} 47--55.

\bibitem{Shiga1981}
{\sc Shiga, T.} (1981).
Diffusion processes in population genetics.
{\em J. Math. Kyoto Univ.} {\bf 21,} 133--151.

\bibitem{Tavare2004}
{\sc Tavar\'{e}, S.} (2004).
Ancestral inference in population genetics. In {\em Lectures on Probability 
Theory and Statistics. Ecole d'Et\'{e} de Probabilit\'{e}s de Saint-Flour 
XXXI--2001.} ed. J. Picard. Lecture Notes in Mathematics, 1837, 1--188. 
Springer-Verlag, New York. 










%

\end{thebibliography}
\end{document}